\documentclass[11pt]{article} 
\usepackage[text={6.5in,9in}]{geometry}
\usepackage{amsmath,latexsym,amssymb,color,amsthm}
\usepackage{ifthen,graphics,epsfig}
\usepackage{color}
\usepackage{xspace}
\usepackage{enumitem}
\usepackage{url}
\usepackage{amssymb}
\usepackage{pifont}
\usepackage{authblk}
\usepackage{booktabs}
\usepackage{wrapfig}
\usepackage[colorlinks=true,citecolor=gray]{hyperref}
\usepackage{caption}
\usepackage{subcaption}
\usepackage[table]{xcolor}

\usepackage{float}
\newfloat{algorithm}{thp}{lop}
\floatname{algorithm}{Algorithm}
\newcommand{\Xomit}[1]{}

\usepackage{marginnote}
\usepackage[table]{xcolor}
\usepackage{todonotes}
\usepackage{menukeys}

\newcommand{\remove}[1]{}

\newlength {\squarewidth}

\newtheorem{theorem}{Theorem}
\newtheorem{lemma}{Lemma}

\newcommand{\toto}{xxx}

\newcounter{linecounter}
\newcommand{\linenumbering}{\ifthenelse{\value{linecounter}<10}{(0\arabic{linecounter})}{(\arabic{linecounter})}}
\renewcommand{\line}[1]{\refstepcounter{linecounter}\label{#1}\linenumbering}
\newcommand{\resetline}[1]{\setcounter{linecounter}{0}#1}
\renewcommand{\thelinecounter}{\ifnum \value{linecounter} > 9\else 0\fi \arabic{linecounter}}

\newcommand{\aux}[2]{\textsc{aux}[#1](#2)}
\newcommand{\auxset}[2]{\textsc{auxset}[#1](#2)}
\newcommand{\auxsc}[1]{\textsc{aux}(#1)}

\newcommand{\supportcoin}{\mathit{support\_coin}}
\newcommand{\msg}{\textsc{tag}}
\newcommand{\stage}{\textsc{stage}}
\newcommand{\bval}{\textsc{b\_val}}
\newcommand{\sval}{\textsc{s\_val}}
\newcommand{\est}{\textsc{est}}

\newcommand{\shouldbroadcast}{\mathit{should\_broadcast}}

\newcommand{\propose}[1]{\textsc{propose}(#1)}

\newcommand{\send}{\mathit{\sf send}}
\newcommand{\sto}{\mathit{\sf to}}
\newcommand{\receive}{\mathit{\sf receive}}
\newcommand{\broadcast}{\mathit{\sf broadcast}}

\newcommand{\ttrue}{\mathit{\tt true}}
\newcommand{\tfalse}{\mathit{\tt false}}

\newcommand{\wait}{{\sf wait\_until}}

\newcommand{\bvbroadcast}{{\sf BV\_broadcast}}
\newcommand{\sbroadcast}{{\sf S\_broadcast}}
\newcommand{\sbvbroadcast}{{\sf SBV\_broadcast}}

\newcommand{\BAMP}{{\cal BAMP}_{n,t}}

\newcommand{\binvalues}{\mathit{bin\_values}}
\newcommand{\svalue}{\mathit{s\_value}}
\newcommand{\binptr}{\mathit{bin\_ptr}}
\newcommand{\view}{\mathit{view}}

\newcommand{\decide}{{\sf decide}}

\title{Two More Algorithms for Randomized Signature-Free Asynchronous Binary Byzantine Consensus with $t < n/3$ and $O(n^2)$ Messages and
  $O(1)$ Round Expected Termination}

\author[1]{Tyler Crain}
\affil[1]{tcrainwork@gmail.com}

\begin{document}

\maketitle

\begin{abstract}
  This work describes two randomized, asynchronous, round based, Binary Byzantine faulty tolerant
  consensus algorithms based on the algorithms of~\cite{MMR14} and \cite{MMR15}.
  Like the algorithms of~\cite{MMR14} and \cite{MMR15} they do not use signatures, use $O(n^2)$ messages per round
  (where each message is composed of a round number and a constant number of bits),
  tolerate up to one third failures, and have expected termination in constant number of rounds.

  The first, like~\cite{MMR15}, uses a weak common coin (i.e. one that can return different
  values at different processes with a constant probability) to ensure termination.
  The algorithm consists of $5$ to $7$ message broadcasts per round.
  An optimization is described that reduces this to $4$ to $5$ broadcasts per round
  for rounds following the first round.
  Comparatively,~\cite{MMR15} consists of $8$ to $12$ message broadcasts per round.
  
  The second algorithm, like~\cite{MMR14}, uses a strong common coin
  (i.e. one that returns the same value at all non-faulty processes)
  for both termination and correctness.
  Unlike~\cite{MMR14}, it does not require a fair scheduler to ensure termination.
  Furthermore, the algorithm consists of $2$ to $3$ message broadcasts for the first round
  and $1$ to $2$ broadcasts for the following rounds,
  while~\cite{MMR15} consists of $2$ to $3$ broadcasts per round.  
\end{abstract}













\section{Introduction and related work.}
Binary byzantine consensus concerns the problem of getting a set of distinct processes
distributed across a network to agree on a single binary value $0$ or $1$ where processes
can fail in arbitrary ways.
It is well known that this problem is impossible in an asynchronous network with at least
one faulty process~\cite{FLP85}. To get around this, algorithms can employ
randomization~\cite{A03, BO83, B87, BT83, CR93, FP90, KS16, MMR14, MMR15, PCR14, R83, T84},
or rely on an additional synchrony assumption~\cite{DDS87, DLS88}.
Randomized algorithms largely rely on the existence of a local or common random coin.
The output of local coin is only visible to an individual process, while the output of a common
coin is visible to all processes, but only once a threshold of processes have participated in computing the coin.
A strong common coin is one that outputs the same value at all processes while a weak one
may output different values at different processes with a fixed probability~\cite{R83}.

This work presents two algorithms that do not use signatures, use $O(n^2)$ messages per round
(where each message is composed of a round number and a constant number of bits),
tolerate up to one third failures (a well know lower bound~\cite{LSP82}),
and have expected termination in a constant number of rounds.

The first algorithm uses a weak common coin and is based on the BV-Broadcast and
SBV-Broadcast abstractions of~\cite{MMR15}.
It consists of two invocations of SBV-Broadcast per round with an additional normal
message broadcast in between.
In the algorithm processes perform between $5$ and $7$ message broadcasts per round.
An optimization is described that reduces this to $5$ to $6$ broadcasts for the first round and
$4$ to $5$ broadcasts for following rounds.
Comparatively, the algorithm in~\cite{MMR15} that solves Byzantine Binary Consensus with
the same guarantees and uses between $8$ and $12$ message broadcasts per round.

The second algorithm uses a strong common coin which allows a reduction of the number of messages
processes broadcast per round to between $2$ to $3$ for the first round and $1$ to $2$ for all following rounds.
In each round processes either decide the same value as the output of the common coin for that round
or continue to the next round without deciding.
It is based on the algorithm of~\cite{MMR14} that solves Byzantine Binary Consensus with
the same guarantees, but requires a fair scheduler to ensure termination and uses
between $2$ and $3$ message broadcasts per round.
The algorithms of~\cite{C220,C20,CGLR18} all follow simmilar designs except with different
assumptions, where~\cite{C20,C220} uses signatures with~\cite{C220} relying on a
synchrony assumption for terminaton, and~\cite{CGLR18} does not use signatures but also
relies on a synchrony assumption for termination.
Note that~\cite{M18} presents a modification of~\cite{MMR14} that uses an additional message
broadcast to remove the fair scheduler requirement,
bringing up the number of broadcasts per round to between $3$ and $4$.

Note that one of the notable properties of the algorithms presented in this work, as well
as those in~\cite{MMR14} and~\cite{MMR15} is that they do not use signatures within
the description of the consensus.
Although it is common knowledge that common coins are implemented using signatures,
such as the efficient coin of~\cite{CKS05},
and that encrypted communication is often setup using signatures, such as with TLS,
signatures do have disadvantages such as that they require expensive computations and increase message
size by a non-insignificant number of bytes.
Therefore, while in practice these algorithms will still require the use of signatures,
not using signatures within the consensus remains beneficial.
Still, many asynchronous binary consensus algorithms exist that use signatures
and signatures have many advantages.
Apart from the advantages provided directly by signatures themselves such as irrefutable,
algorithms using them can be less complex, for example in the algorithm of~\cite{CKS05}, which uses
a weak common coin, processes perform $2$ broadcasts per round and in the algorithm of~\cite{C20}, which
uses a strong common coin, processes broadcast $1$ message per round.
Given the different tradeoffs of using signatures or not there are many randomized algorithms
that use them~\cite{C20, R83, T84, CR93, CKS05} and many that do not~\cite{B87, BG93, C20, FMR07, ST87}.



While the binary consensus problem only allows process to agree on a single binary value,
there exist many reductions to multi-value consensus~\cite{MR17, MRT00, TC84, ZC09} allowing processes to agree on arbitrary values.
Furthermore many algorithms~\cite{BSA14, CL02} exists that solve multi-value consensus directly
through the use of types of synchrony assumptions to ensure termination.
Additionally, algorithms exists that make many different assumptions about the model
such as synchrony~\cite{FM97}, different fault models~\cite{LVCQV16, MA06, PSL80},
solve different definitions of consensus~\cite{NCV05}, and so on.

\section{A Byzantine Computation Model.}
\label{sec:model}

This section describes the assumed computation model.

\paragraph{Asynchronous processes.}
The system is made up of a set $\Pi$ of $n$ asynchronous sequential processes,
namely $\Pi = \{p_1,\ldots,p_n\}$; $i$ is called the ``index'' of $p_i$. 
``Asynchronous'' means that each process proceeds at its own speed,
which can vary with time and remains unknown to the other processes.
``Sequential'' means that a process executes one step at a time.
This does not prevent it from executing several threads with an appropriate
multiplexing. 
%
Both notations
$i\in Y$ and $p_i\in Y$ are used to say that $p_i$ belongs to the set $Y$.

\paragraph{Communication network.}
\label{sec:basic-comm-operations}
The processes communicate by exchanging messages through
an asynchronous reliable point-to-point network. ``Asynchronous''  means that
there is no bound on message transfer delays, but these delays are finite.
``Reliable'' means that the network does not lose, duplicate, modify, or
create messages. ``Point-to-point'' means that any pair of processes
is connected by a bidirectional channel.
%
A process $p_i$ sends a message to a process $p_j$ by invoking the primitive 
``$\send$ {\sc tag}$(m)$ $\sto~p_j$'', where {\sc tag} is the type
of the message and $m$ its content. To simplify the presentation, it is
assumed that a process can send messages to itself. A process $p_i$ receives 
a message by executing the primitive ``$\receive()$''.
The macro-operation $\broadcast$ {\sc tag}$(m)$ is  used as a shortcut for
``{\bf for each} $p_i \in \Pi$  {\bf do} $\send$ {\sc tag}$(m)$ $\sto~p_j$
{\bf end for}''. 

\paragraph{Failure model.}
Up to $t$ processes can exhibit a {\it Byzantine} behavior~\cite{PSL80}.
 A Byzantine process is a process that behaves
arbitrarily: it can crash, fail to send or receive messages, send
arbitrary messages, start in an arbitrary state, perform arbitrary state
transitions, etc. Moreover, Byzantine processes can collude 
to ``pollute'' the computation (e.g., by sending  messages with the same 
content, while they should send messages with distinct content if 
they were non-faulty). 
A process that exhibits a Byzantine behavior is called {\it faulty}.
Otherwise, it is {\it non-faulty}.  
%
Moreover, it is assumed that the Byzantine processes do not
fully control the network in that they can not corrupt the messages sent by 
non-faulty  processes.
Byzantine processes can control the network by modifying
the order in which messages are received, but they cannot
postpone forever message receptions.  


\paragraph{A Common Coin.}
The model is enriched with the same \emph{common coin} (CC) as in~\cite{MMR15}
that was originally defined in~\cite{R83}.
The common coin outputs a binary value at each non-faulty process
for each round.
All non-faulty processes output $0$ in round $r$ with probability $1/d$
and output $1$ in round $r$ with probability $1/d$.
Non-faulty processes output different values in round $r$ with probability $(d-2)/d$,
where $d \geq 2$ is a known constant.
The output of the coin is revealed by calling a function {\sf random}() provided by
a random oracle.
The output of the coin is unpredictable and random and its output is only revealed
for a round $r$ once at least one non-faulty process has called {\sf random}() in that round
(i.e. faulty processes cannot compute the output of the coin entirely themselves).

\paragraph{A Strong ($t+1$) Common Coin.}
A \emph{strong common coin} (SCC) is defined as a common coin that has $d = 2$ meaning that in every round
all non-faulty processes receive the same output from the common coin.
Furthermore $t+1$ means the output of the coin for round $r$ is not revealed until
at least $t+1$ non-faulty processes have called called {\sf random}() in that round.

\paragraph{Notations.}

\begin{itemize}
\item The acronym ${\BAMP}[\emptyset]$ is used to denote the 
  basic Byzantine Asynchronous Message-Passing computation model;
  $\emptyset$ means that there is no additional assumption. 
\item The basic computation model strengthened with the additional constraint $t<n/3$
  is denoted ${\BAMP}[t<n/3]$.
\item ${\BAMP}[t<n/3]$ enriched with the common coin is denoted ${\BAMP}[t<n/3,CC]$.
\item ${\BAMP}[t<n/3]$ enriched with the strong ($t+1$) common coin is denoted ${\BAMP}[t<n/3,SCC]$.
\end{itemize}


Before presenting the algorithms, the BV-Broadcast and SBV-Broadcast abstractions
from~\cite{MMR15} are recalled.
Note that within the consensus algorithms multiple instances of these abstractions may be used,
so to differentiate between these instances they are called with unique tags
(denoted as $\msg$ in the presentation of the abstractions).

\subsection{The BV-Broadcast abstraction from~\cite{MMR15} in ${\BAMP}[t<n/3]$}

\begin{figure*}[ht!]
\centering{
\fbox{
\begin{minipage}[t]{150mm}
\footnotesize
\renewcommand{\baselinestretch}{2.5}
\resetline
\begin{tabbing}
aaaA\=aaA\=aaaA\=aaaaaaaaaA\kill

{\bf opera}\={\bf tion} ${\bvbroadcast}$ $\msg(v_i)$ {\bf is}\\

\line{BV-01} \> $\binvalues_i \leftarrow \emptyset$ \\

\line{BV-02} \> $\broadcast$ $\bval(v_i)$ \\

\line{BV-03} \> {\bf return} $\binvalues_i$
{\it \scriptsize  \hfill \color{gray}{// $\binvalues_i$ has not necessarily obtained its final value when returned}} \\~\\

{\bf when} $\bval(v)$ is received \\

\line{BV-04} \> {\bf if} ($\bval(v)$ received from $(t+1)$ different processes and $\bval(v)$ not yet broadcast) \\

\line{BV-05} \>\> {\bf then} $\broadcast$ $\bval(v)$
{\it \scriptsize  \hfill \color{gray}{// a process echos a value only once}} \\

\line{BV-06} \> {\bf end if}; \\

\line{BV-07} \> {\bf if} ($\bval(v)$ received from $(2t+1)$ different processes) \\

\line{BV-08} \>\> {\bf then} $\binvalues_i \leftarrow \binvalues_i \cup \{v\}$
{\it \scriptsize  \hfill \color{gray}{// local delivery of a value}} \\

\line{BV-09} \> {\bf end if}.

\end{tabbing}
\normalsize
\end{minipage}
}
\caption{ An algorithm implementing BV-broadcast in ${\BAMP}[t<n/3]$ from~\cite{MMR15}.}
\label{algo-BV} 
\vspace{-1em}
}
\end{figure*}
For each instance of BV-Broadcast, each non-faulty process $p_i$ calls the abstraction with
a unique tag $\msg$ for that instance and a binary value as input.
It returns a set of binary values $\binvalues$, which has not necessarily achieved
its final state when returned (i.e. the implementation of the abstraction may add items later).
The abstraction ensures the following properties:

\begin{itemize}
\item BV-Termination. The invocation of $\bvbroadcast()$ by a non-faulty process terminates.
\item BV-Justification. If $p_i$ is non-faulty and $v \in \binvalues_i$, then $v$ has been BV-Broadcast
  by a non-faulty process.
\item BV-Uniformity. If a value $v$ is added to the set $\binvalues_i$ of a non-faulty process $p_i$,
  eventually $v \in \binvalues_j$ at every non-faulty process $p_j$.
\item BV-Obligation. Eventually the set $\binvalues_i$ of a non-faulty process is non empty.
\item BV-Single-value. If all non-faulty processes BV-Broadcast the same value $v$, $v$ is eventually
  added to the set $\binvalues_i$ of each non-faulty process $p_i$.
\end{itemize}

Note that the values input by non-faulty processes for a specific instance need not be from the set $\{0,1\}$,
they can be any values as long as the size of the set of values input by non-faulty processes is between $1$ and $2$.
An algorithm implementing the BV-Broadcast abstraction is presented in Figure~\ref{algo-BV}.
The reader is referred to~\cite{MMR15} for proofs and a description of the code.

\subsection{The SBV-Broadcast abstraction from~\cite{MMR15} in ${\BAMP}[t<n/3]$}

\begin{figure*}[ht!]
\centering{
\fbox{
\begin{minipage}[t]{150mm}
\footnotesize
\renewcommand{\baselinestretch}{2.5}
\resetline
\begin{tabbing}
aaaA\=aaA\=aaaA\=aaaaaaaaaA\kill

{\bf opera}\={\bf tion} $\sbvbroadcast$ $\msg(v_i)$ {\bf is}\\

\line{SBV-01} \> $\binvalues_i \leftarrow \bvbroadcast$ $\msg(v_i)$; \\

\line{SBV-02} \> \wait ($\binvalues_i \neq \emptyset$) \\
{\it \scriptsize  \hfill \color{gray}{// $\binvalues_i$ has not necessarily obtained its final value when the wait terminates}} \\

\line{SBV-03} \> $\broadcast$ $\auxsc{w}$ where $w \in \binvalues_i$ \\

\line{SBV-04} \> \wait ($\exists$ a set $\view_i$ such that its values (i) belong to $\binvalues_i$ and  \\
\>\>   $~~~~~~~$ (ii) come from messages $\auxsc$ received from $(n-t)$ distinct processes); \\

\line{SBV-05} \> {\bf return}($\view_i, \binvalues_i$)



\end{tabbing}
\normalsize
\end{minipage}
}
\caption{ An algorithm implementing SBV-Broadcast in ${\BAMP}[t<n/3]$ from~\cite{MMR15}.}
\label{algo-SBV} 
\vspace{-1em}
}
\end{figure*}

For each instance of SBV-Broadcast, each non-faulty process $p_i$ calls the abstraction with
a unique tag $\msg$ for that instance and a binary value as input.
It returns two sets of binary values, the first being $\view$ and the second being
$\binvalues$.
While the set $\view$ has achieved its final state when returned, $\binvalues$ may not have
(i.e. the implementation of the abstraction may add items later).
The abstraction ensures the following properties:
(Note that the original SBV-Broadcast from~\cite{MMR15} only returned a single set $\view$,
but the algorithms presented in this work will need an additional property ensured
by $\binvalues$.
Also note that this set already exists in the original implementation,
here it is simply returned.)

\begin{itemize}

\item SBV-Termination. The invocation of $\sbvbroadcast$ $\msg()$ by a non-faulty process terminates.

\item SBV-Obligation. The set $\view_i$ returned by a non-faulty process $p_i$ is not empty.

\item SBV-Justification. If $p_i$ is non-faulty and $v \in \view_i$ then a non-faulty
  process called $\sbvbroadcast$ $v$.

\item SBV-Inclusion. If $p_i$ and $p_j$ are non-faulty processes and $\view_i = \{v\}$
  then $v \in \view_j$.

\item SBV-Uniformity. If all non-faulty processes $\sbvbroadcast$ the same value $v$, then
  $\view_i= \{v\}$ at every non-faulty process $p_i$.

\item SBV-Singleton. If $p_i$ and $p_j$ are non-faulty, $[(\view_i=\{v\}) \wedge (view_j = \{w\})] \implies (v = w)$

\end{itemize}

In this work the following additional property is introduced that was not originally included in~\cite{MMR15}:

\begin{itemize}

\item SBV-Binvalues. The set $\binvalues_i$ returned by a non-faulty process $p_i$
  satisfies the BV-Broadcast abstraction where non-faulty processes input the same values
  to this abstraction as they did to the SBV-Broadcast abstraction.
  Furthermore the set $\binvalues_i$ at a
  non-faulty process $p_i$ eventually contains every value
  returned in $view_j$ at every non-faulty process $p_j$.

\end{itemize}

An algorithm implementing the SBV-Broadcast abstraction is presented in Figure~\ref{algo-SBV}.
The reader is referred to~\cite{MMR15} for full proofs and a description of the code.
Here the SBV-Binvalues is proved as it is added in this work.

\begin{lemma}
  The algorithm of Figure~\ref{algo-SBV} implements the SBV-Broadcast abstraction in ${\BAMP}[t<n/3]$.
  \label{lem:sbvbroadcast}
\end{lemma}
\begin{proof}
  All properties other than SBV-Binvalues are proved in~\cite{MMR15}.
  Proof of the SBV-Binvalues property: By line~\ref{SBV-04} the set $\view_i$ returned on
  line~\ref{SBV-04} is the set $\binvalues_i$ returned by the call
  to $\bvbroadcast$ which satisfies the BV-Broadcast abstraction.
  By BV-Uniformity all processes will eventually have the same set
  of values in their set $\binvalues_i$ which is the set returned by SBV-Broadcast,
  the proof then follows.
\end{proof}

\section{Binary Byzantine Consensus.}
\label{sec:byz-consensus}

\subsection{The Binary Consensus Problem.}


%

In the binary consensus problem processes input a value to the algorithm, called their \emph{proposal},
run an algorithm consisting of several rounds,
and eventually output a binary value called their \emph{decision}.
Let $\cal V$ be the set of values that can be proposed.
While  $\cal V$ can contain any number ($\geq 2$) of values
in multi-valued consensus, it contains only two values in binary consensus, 
e.g., ${\cal V} =\{0,1\}$.
Assuming that
each non-faulty process proposes a value, the binary Byzantine consensus (BBC) problem is for 
each of them to
decide on a value in such a way that the following properties are
satisfied:
\begin{itemize}
\item BBC-Termination. Every non-faulty process eventually decides on a value.
\item BBC-Agreement.   No two non-faulty processes decide on different values.
\item BBC-Validity.  If all non-faulty processes propose the same value, no
other value can be decided.
\end{itemize}


\subsection{A Safe and Live Binary Byzantine Consensus Algorithm in ${\BAMP}[t<n/3,CC]$.}
\label{ssec:live-bbc}

This section presents a Binary Byzantine Consensus algorithm using
SBV-Broadcast in addition to a weak common coin.

\paragraph{Message types.}
The following message types are used by the consensus (in addition to those used by SBV-Broadcast).
\begin{itemize}
\item $\auxset{r}{s}$. An {\sc auxset} message contains a round number $r$ and a set of binary values $s$.
\end{itemize}

\paragraph{Local variables.}
The following local variables are used at each process.
\begin{itemize}
\item $r_i$. The current round at process $i$.
\item $\view_i[]$. A map of sets of binary values at process $i$ indexed by a round and an integer $0-2$.
\item $\binvalues_i[]$. A map of sets of binary values at process $i$ indexed by a round.
\item $est_i$. The current binary estimate at process $i$.
\end{itemize}








\begin{figure*}[ht!]
\centering{
\fbox{
\begin{minipage}[t]{150mm}
\footnotesize
\renewcommand{\baselinestretch}{2.5}
\resetline
\begin{tabbing}
aaaA\=aaA\=aaaA\=aaaaaaaaaA\kill

{\bf opera}\={\bf tion} $\propose{v_i}$ {\bf is}\\

\line{BBC3-01} \> $est_i \leftarrow v_i; r_i \leftarrow 0$; \\

{\bf repeat forever} \\

\line{BBC3-02} \> $r_i \leftarrow r_i + 1$; \\

\line{BBC3-03} \> $(\view_i[r_i,0], \binvalues_i[r_i]) \leftarrow \sbvbroadcast$ $\stage[r_i,0](est_i)$; \\
{\it \scriptsize  \hfill \color{gray}{// $\binvalues_i$ has not necessarily obtained its final value}} \\

\line{BBC3-04} \> $\broadcast$ $\auxset{r_i}{\view_i[r_i,0]}$
{\it \scriptsize  \hfill \color{gray}{// broadcast a set of binary values}} \\

\line{BBC3-05} \> \wait ($\exists$ a set $\view_i[r_i,1]$ such that its values (i) belong to $\binvalues_i[r_i]$ and  \\
\>\>   $~~~~~~~$ (ii) come from messages $\auxset{r_i}{}$ received from $(n-t)$ distinct processes); \\

\line{BBC3-06} \> {\bf if} ($view_i[r_i,1] = \{w\}$) \\

\line{BBC3-07} \>\> {\bf then} $est_i \leftarrow w$ \\

\line{BBC3-08} \>\> {\bf else} $est_i \leftarrow \bot$ \\

\line{BBC3-09} \> {\bf end if} \\

\line{BBC3-10} \> $(\view_i[r_i,2], \_) \leftarrow \sbvbroadcast$ $\stage[r_i,1](est_i)$; \\

\line{BBC3-11} \> $s \leftarrow$ {\sf random}(); \\

\line{BBC3-12} \> {\bf case} ($\view_i[r_i,2] = \{v\} \wedge v \neq \bot$) {\bf then} $est_i \leftarrow v; \decide(v)$ if not yet done \\

\line{BBC3-13} \>\> $~~~$ ($\view_i[r_i,2] = \{v, \bot\}$) $~~~~~~$ {\bf then} $est_i \leftarrow v$ \\

\line{BBC3-14} \>\> $~~~$ ($\view_i[r_i,2] = \{\bot\}$) $~~~~~~~~~$ {\bf then} $est_i \leftarrow s$ \\

\line{BBC3-15} \> {\bf end case} \\

{\bf end repeat}.

\end{tabbing}
\normalsize
\end{minipage}
}
\caption{ An algorithm implementing binary consensus in ${\BAMP}[t<n/3,CC]$.}
\label{algo-BBC} 
\vspace{-1em}
}
\end{figure*}

\subsubsection{Algorithm Description.}
Figure~\ref{algo-BBC} presents the algorithm.
Non-faulty process call $\propose$ with an initial binary proposal.
Line~\ref{BBC3-01} initializes the processes' estimate to its proposal and the round to $0$.
Non-faulty processes then repeat lines~\ref{BBC3-02}-\ref{BBC3-15} for each round.

First the round number is incremented on line~\ref{BBC3-02}.
Non-faulty processes then call $\sbroadcast$ with tag $\stage[r_i,0]$ and input $est_i$.
By SBV-Singleton this call will output $\view_i[r_i,0]$ with a single unique value $v$
or both binary values at all non-faulty processes.

The idea behind this SBV-Broadcast is that (i) if $\view_i[r_i,0] = \{v\}$ at least $1$ non-faulty process
then the remaining code if the round will ensure only $v$ may be decided
and additionally, if the output of the coin is also $v$
at all all non-faulty processes, then they will
set their estimate to $v$ and decide by the following round.

Otherwise, (ii) if $\view_i[r_i,0] = \{0,1\}$ at non-faulty processes before the value
of the coin is revealed then lines~\ref{BBC3-04}-\ref{BBC3-05} are important to
help non-faulty processes reach a decision.
On line~\ref{BBC3-04}, non-faulty processes broadcast $\auxset{r_i}{\view_i[r_i,0]}$
and on line~\ref{BBC3-05} wait for ($n-t$) of these messages from distinct processes where the values
are contained in $\binvalues_i$.
The set $\view_i[r_i,1]$ is then computed as the set of these values.
If $\view_i[r_i,1]$ contains a single value then a process sets its estimate to this value
otherwise it sets its estimate to $\bot$.
If ($t+1$) non-faulty processes have $\view[i,0] = \{0,1\}$,
and broadcast $\auxset{r}{\{0,1\}}$, then given $t < n/3$,
any set of ($n-t$) $\auxset{r}{}$ messages will contain at least one of these messages meaning
all non-faulty process will set their estimate to $\bot$.
The remaining code of the round will then ensure all processes set their estimate to the output of the coin and decide
in the following round if the coin outputs the same value at all non-faulty processes.

Non-faulty processes then make a second call to $\sbvbroadcast$ with tag $\stage[r,1]$ and their current estimate as input.
This returns the set $\view[r,2]$ (line~\ref{BBC3-10}).
A call to {\sf random}() is then made on line~\ref{BBC3-11}.
The purpose of this call to $\sbvbroadcast$ is to ensure that through SBV-Uniformity
if a non-faulty process decides then all non-faulty processes
set their estimates to the same value and decide in the following round.
For this, if $\view[r,2]$ is a single value then this value is decided and is set as the processes' estimate (line~\ref{BBC3-12}),
otherwise if view contains both a binary value and $\bot$ then the binary value is set to the processes' estimate (line~\ref{BBC3-13}),
otherwise $\view[r,2] = \{\bot\}$ and the process sets its output to the value of the coin (line~\ref{BBC3-14}).
Non-faulty processes then continue on to the next round.

\subsubsection{Proofs.}

This section shows that the algorithm of Figure~\ref{algo-BBC} solves Binary Byzantine consensus
and terminates in an expected constant number of rounds.

\begin{lemma}
  At the start of every round each non-faulty process has a binary estimate proposed by a non-faulty process.
  \label{lem:alwaysbin}
\end{lemma}
\begin{proof}
  The initial estimate at non-faulty processes is set to its binary
  proposal $v$ on line~\ref{BBC3-01}.
  The estimate is then input to $\sbvbroadcast$ $\stage[r,0]$ on line~\ref{BBC3-03}.
  SBV-Justification and SBV-Binvalues then ensure both $\view[r,0]$ and $\view[r,1]$
  only contain binary values proposed by non-faulty processes.
  The estimate is then set to either a binary value from $\view[r,1]$ or $\bot$ if both $0$ and $1$ are contained
  in $\view[r,1]$ on lines~\ref{BBC3-07}-\ref{BBC3-08}.
  The estimate is then input to $\sbvbroadcast$ $\stage[r,1]$ on line~\ref{BBC3-10}.
  By SBV-Justification $\view[r,2]$ can only contain binary values or $\bot$ where $\bot$ means both $1$ and $0$
  were proposed by non-faulty processes.
  On lines~\ref{BBC3-12}-\ref{BBC3-14} the estimate is then set to either the binary value
  in $\view[r,2]$, or the value of {\sf random}() if $\bot \in \view[r,2]$, where in either case the value
  must have been proposed by a non-faulty process.
  From this the next round is then started with a binary estimate proposed by a non-faulty process and the proof
  is the same for all following rounds.
\end{proof}

\begin{lemma}
  Only $\bot$ or binary values can be input to $\sbvbroadcast$ $\stage[r,1]$ on line~\ref{BBC3-10} by non-faulty processes.
  The binary values input have been proposed by non-faulty processes or if $\bot$ is input
  then both $0$ and $1$ have been proposed by non-faulty processes.
  \label{lem:botbin}
\end{lemma}
\begin{proof}
  By Lemma~\ref{lem:alwaysbin} all non-faulty processes start each round with a binary estimate
  proposed by a non-faulty process which
  is input to $\sbvbroadcast$ $\stage[r,0]$ on line~\ref{BBC3-03}.
  Now by SBV-Binvalues, $\binvalues$ will only contain binary values at non-faulty processes.
  Thus $\view[r,1]$ will only contain binary values (line~\ref{BBC3-05}) proposed by non-faulty processes
  and lines~\ref{BBC3-06}-\ref{BBC3-08} will ensure only $\bot$ (if both $0$ and $1$ have been proposed by a non-faulty process)
  or a binary value is input to $\sbvbroadcast$ $\stage[r,1]$ on line~\ref{BBC3-10} by non-faulty processes.
\end{proof}

\begin{lemma}
  Non-faulty processes will only decide binary values proposed by non-faulty processes.
  \label{lem:nfdec}
\end{lemma}
\begin{proof}
  From line~\ref{BBC3-12} non-faulty processes can only decide binary values output by
  $\sbvbroadcast$ $\stage[r,1]$ (lines~\ref{BBC3-10}-\ref{BBC3-12}).
  By Lemma~\ref{lem:botbin} and SBV-Justification these values must have been proposed by non-faulty processes.
\end{proof}

\begin{lemma}
  If in a round $r$, $\view[r,0]$ contains a single binary value $w$ at a non-faulty process,
  then all non-faulty processes either input $w$ or $\bot$ to $\sbvbroadcast$ $\stage[r,1]$.
  \label{lem:view0single}
\end{lemma}
\begin{proof}
  If $\view[r,0] = \{w\}$ at a non-faulty process, then by Lemma~\ref{lem:alwaysbin} and
  SBV-Singleton either $\view[r,0] = \{w\}$ or $\view[r,0] = \{0,1\}$
  at all non-faulty processes.
  Now non-faulty processes will either broadcast either $\auxset{r}{\{w\}}$ or $\auxset{r}{\{0,1\}}$ on line~\ref{BBC3-04}.
  Given $t < n/3$ no non-faulty process will receive $n-t$ $\auxset{r}{\{\neg w\}}$ messages from distinct processes
  and as a result $view[r,1]$ will be set to either $\{w\}$ or $\{0,1\}$ at all non-faulty processes on line~\ref{BBC3-05}.
  Now by lines~\ref{BBC3-06}-\ref{BBC3-08} all non-faulty processes will set their estimate to $w$ or $\bot$
  which is then input to $\sbvbroadcast$ $\stage[r,1]$.
\end{proof}

\begin{lemma}
  For any round $r$, if a non-faulty process inputs a binary value $v$ to $\sbvbroadcast$ $\stage[r,1]$ on line~\ref{BBC3-10}
  then all non-faulty processes either input $v$ or $\bot$ to $\sbvbroadcast$ $\stage[r,1]$.
  \label{lem:binall}
\end{lemma}
\begin{proof}
  By Lemma~\ref{lem:alwaysbin}, Lemma~\ref{lem:view0single} and SBV-Justification this holds true for all cases
  except where all non-faulty processes have $\view[r,0] = \{0,1\}$.
  In this case and given $t < n/3$ all non-faulty processes will receive at least $t+1$ $\auxset{r}{\{0,1\}}$
  messages from distinct processes and as a result $\view[r,1] = \{0,1\}$ at all non-faulty processes (line~\ref{BBC3-05}).
  All non-faulty processes then set their estimate to $\bot$ on line~\ref{BBC3-08} which is then input to
  $\sbvbroadcast$ $\stage[r,1]$ and the lemma holds.
\end{proof}

\begin{lemma}
  Non-faulty processes complete each round.
  \label{lem:complete}
\end{lemma}
\begin{proof}
  By Lemma~\ref{lem:alwaysbin} all non-faulty processes start each round with a binary estimate which
  is input to $\sbvbroadcast$ $\stage[r,0]$ on line~\ref{BBC3-03}.  
  By SBV-Termination all processes will complete the call to $\sbvbroadcast$ $\stage[r,0]$ line~\ref{BBC3-03} and broadcast
  $\auxset{r}{\view[r,0]}$ on line~\ref{BBC3-04}.
  By Lemma~\ref{lem:alwaysbin} and SBV-Binvalues the set $\view[r,0]$ will contain binary values from $\binvalues$.
  Given this, all non-faulty processes will then receive $(n-t)$ $\auxset{r}{s}$ messages from distinct
  processes where $s \in \binvalues$ and complete line~\ref{BBC3-05}.
  Non-faulty processes will then call and complete $\sbvbroadcast$ $\stage[r,1]$ on line~\ref{BBC3-10}
  by SBV-Termination.
  All non-faulty processes will then call {\sf random}() on line~\ref{BBC3-11} which will return a binary value.
  
  From Lemma~\ref{lem:binall} which ensure no two distinct binary values will be
  input to $\sbvbroadcast$ $\stage[r,1]$ on line~\ref{BBC3-10} and from SBV-Justification
  $\view[r_i,2]$ will contain either a single binary value or a single binary value and $\bot$,
  which will then match a valid case on lines~\ref{BBC3-12}-\ref{BBC3-15}.
  Processes will then continue to the next round where the same proof construction applies.
\end{proof}

\begin{lemma}
  If all non-faulty processes start a round $r$ with the same binary estimate $v$
  then they all decide $v$ in round $r$ (if not already done) and never
  decide a different value in following rounds.
  \label{lem:est}
\end{lemma}
\begin{proof}
  By definition of the lemma, in round $r$
  all non-faulty processes input the same estimate $v$ to $\sbvbroadcast$ $\stage[r,0]$ on line~\ref{BBC3-03}.
  By SBV-Uniformity all non-faulty processes have $\view[r,0] = \{v\}$ and
  by SBV-Binvalues only $v$ can be contained in $\binvalues$ at non-faulty processes.
  All non-faulty processes then broadcast $\auxset{r}{\{v\}}$ on line~\ref{BBC3-04} and set
  $view[r,1]$ to $\{v\}$ on line~\ref{BBC3-05}. Following this $est$ is set to $v$ on line~\ref{BBC3-07}.
  All non-faulty processes then input $v$ to $\sbvbroadcast$ $\stage[r,1]$ on line~\ref{BBC3-10}
  and by SBV-Uniformity have $\view[r,2] = \{v\}$.
  All non-faulty processes then set $est$ to $v$ and decide $v$ (line~\ref{BBC3-12}) if not yet done and
  then start round $r+1$ with the same binary estimate $v$ for which the same proof holds.
\end{proof}

\begin{lemma}
  No two non-faulty processes decide different values.
  \label{lem:samedec}
\end{lemma}
\begin{proof}
  Let $r$ be the first round where a non-faulty process decides.
  By line~\ref{BBC3-12} for a process to decide, a single binary value $v$ must have been returned from
  $\sbvbroadcast$ $\stage[r,1]$ on line~\ref{BBC3-10}.
  Given this and SBV-Justification a non-faulty process must have input $v$ to
  $\sbvbroadcast$ $\stage[r,1]$ on line~\ref{BBC3-10}.
  Furthermore by Lemma~\ref{lem:binall} only $v$ or
  $\bot$ could have been input by non-faulty processes to $\sbvbroadcast$ $\stage[r,1]$.
  Now given SBV-Justification and SBV-Singleton $\view[r,2]$
  must be either $\{v\}$ or $\{v,\bot\}$ at all non-faulty processes.
  Thus on lines~\ref{BBC3-12}-\ref{BBC3-13} all non-faulty processes will
  set $est = v$ and either decide $v$ or not decide in round $r$.
  Lemma~\ref{lem:est} then ensures that in following rounds all non-faulty processes
  decide only $v$.
\end{proof}

\begin{lemma}
  \label{lem:term}
  Non-faulty processes decide in expected $O(1)$ rounds.
\end{lemma}
\begin{proof}
  By definition, the value output by the coin will not be revealed
  in a round $r$ until at least $1$ non-faulty processes has called {\sf random}(),
  i.e. a non-faulty process has reached line~\ref{BBC3-11}.
  For a non-faulty process to reach this line it must have received $n-t$ messages from distinct processes on
  line~\ref{BBC3-05} meaning at least $t+1$ non-faulty processes have reached line~\ref{BBC3-04} in round $r$
  before the coin is revealed.
  Consider the following two possible cases at the point where the $t+1$th non-faulty process reaches this line,
  and before the value of the coin for round $r$ is revealed (note that by Lemma~\ref{lem:complete} all
  non-faulty processes will eventually reach this line).

  \begin{itemize}
  \item First assume that at least one of the $t+1$ non-faulty processes had $\view[r,0]$ returned
    from $\sbvbroadcast$ $\stage[r,0]$ (line~\ref{BBC3-03}) containing a single binary value $w$.
    Now by Lemma~\ref{lem:view0single}
    all non-faulty processes will input either $\bot$ or $w$ into $\sbvbroadcast$ $\stage[r,1]$
    (line~\ref{BBC3-10}).
    Then by SBV-justification all non-faulty processes will have either $\{w\}$ or $\{w,\bot\}$ or $\{\bot\}$
    as $\view[r,2]$.
    Now if the output of the coin is $w$ at all non-faulty processes then by lines~\ref{BBC3-12}-\ref{BBC3-14} they all
    will set their estimates to $w$.
    Thus with probability of at least $1/d$ all non-faulty processes will set their estimate to the same binary value
    and by Lemma~\ref{lem:est} will decide by the following round.
    
  \item Otherwise, given Lemma~\ref{lem:alwaysbin} and SBV-Justification
    the $t+1$ non-faulty processes must have had $\view[r,0]$ returned
    from $\sbvbroadcast$ $\stage[r,0]$ (line~\ref{BBC3-03}) as the set $\{0,1\}$
    and broadcast $\auxset{r}{\{0,1\}}$ on line~\ref{BBC3-04}.
    Now given $t < n/3$, at all non-faulty processes on line~\ref{BBC3-05} any set of messages from ($n-t$) distinct processes
    will contain at least one $\auxset{r}{\{0,1\}}$ message.
    With this and by SBV-Binvalues all non-faulty processes will eventually have $\view[r,1] = \{0,1\}$
    and will set $est \leftarrow \bot$ on line~\ref{BBC3-08}.
    All non-faulty processes will then input $\bot$ into $\sbvbroadcast$ $\stage[r,1]$ on line~\ref{BBC3-10} and will return $\{\bot\}$
    as $\view[r,2]$ by SBV-Uniformity. All non-faulty processes will then set their estimates to
    the output of the coin on line~\ref{BBC3-14}.
    Thus with probability of at least $1/d$ all non-faulty processes will set their estimate to the same binary value
    and by Lemma~\ref{lem:est} will decide by the following round.
  \end{itemize}

  In both cases in round $r$ non-faulty processes reach a state where they will reach a decision with probability
  of at least $1/d$, or decision
  is ensured with probability $1 - \prod_{r=1}^{\infty} 1/d = 1$.
  From this, the expected number of rounds to reach a state from which a decision is ensured is
  $\sum_{r=1}^{\infty} \frac{r}{d}(1-\frac{1}{d})^{r-1} = d$ with all non-faulty processes deciding by round $d+1$.
\end{proof}

\begin{theorem}
  The algorithm presented in Figure~\ref{algo-BBC} solves the Binary consensus problem in ${\BAMP}[t<n/3,CC]$.
\end{theorem}
\begin{proof}
  First recall the definition of Binary Byzantine Consensus.
  \begin{itemize}
  \item BBC-Termination. Every non-faulty process eventually decides on a value.
  \item BBC-Agreement.   No two non-faulty processes decide on different values.
  \item BBC-Validity.  If all non-faulty processes propose the same value, no
    other value can be decided.
  \end{itemize}
  BBC-Termination is ensured by Lemma~\ref{lem:term}.
BBC-Agreement and BBC-Validity are ensured by Lemmas~\ref{lem:samedec} and~\ref{lem:nfdec} respectively.
\end{proof}

\paragraph{Message broadcasts.}
As defined in~\cite{MMR15}, the implementation of the SBV-Broadcast abstraction
consists of $2$ to $3$ message broadcasts.
The algorithm of Figure~\ref{algo-BBC} consists of a single message broadcast
and two instances of SBV-Broadcasts, or $5$ to $7$ message broadcasts.
Section~\ref{sec:opt} describes a way to reduce this to $4$ to $5$ broadcasts.

\subsection{The S-Broadcast abstraction in ${\BAMP}[t<n/3]$}

The following sections present a Binary Byzantine consensus algorithm based on the
Algorithm of~\cite{MMR15} that uses
a strong ($t+1$) common coin for correctness and termination.
Each round of the algorithm of~\cite{MMR15} consists of a
call to BV-Broadcast followed by a normal message broadcast.
Unfortunately here, in order to remove the fair scheduler requirement
of~\cite{MMR15} without increasing the number of message broadcasts
the BV-Broadcast abstraction can no longer be used.
Instead an S-Broadcast abstraction is introduced that can be thought
of as breaking the BV-Broadcast abstraction into two separate instances
of S-Broadcast, one for each binary value.
Note that it follows a classical approach of broadcast/echo
used in many similar abstractions.

\begin{figure*}[ht!]
\centering{
\fbox{
\begin{minipage}[t]{150mm}
\footnotesize
\renewcommand{\baselinestretch}{2.5}
\resetline
\begin{tabbing}
aaaA\=aaA\=aaaA\=aaaaaaaaaA\kill

{\bf opera}\={\bf tion} ${\sbroadcast}$ $\msg(v_i, \shouldbroadcast_i)$ {\bf is}\\

\line{SBC-01} \> $\svalue_i \leftarrow \tfalse$ \\

\line{SBC-02} \> {\bf if} ($\shouldbroadcast_i = \ttrue$) {\bf then} $\broadcast$ $\msg,\sval(v_i)$ {\bf end if} \\

\line{SBC-03} \> {\bf return} $\svalue_i$ \\
{\it \scriptsize  \hfill \color{gray}{// the Boolean pointed to by $\svalue_i$ has not necessarily obtained its final value when returned}} \\~\\

{\bf when} $\msg,\sval(v)$ is received where ($v = v_i$) \\

\line{SBC-04} \> {\bf if} ($\msg,\sval(v)$ received from $(t+1)$ different processes and $\msg,\sval(v)$ not yet broadcast) \\

\line{SBC-05} \>\> {\bf then} $\broadcast$ $\msg,\sval(v)$
{\it \scriptsize  \hfill \color{gray}{// a process echos a value only once}} \\

\line{SBC-06} \> {\bf end if}; \\

\line{SBC-07} \> {\bf if} ($\msg,\sval(v)$ received from $(2t+1)$ different processes) \\

\line{SBC-08} \>\> {\bf then} $\svalue_i \leftarrow \ttrue$
{\it \scriptsize  \hfill \color{gray}{// local delivery of a value}} \\

\line{SBC-09} \> {\bf end if}.

\end{tabbing}
\normalsize
\end{minipage}
}
\caption{ An algorithm implementing S-broadcast in ${\BAMP}[t<n/3]$.}
\label{algo-SBC} 
\vspace{-1em}
}
\end{figure*}

The S-Broadcast abstraction takes as input a value $v$, and a Boolean value $\shouldbroadcast$.
It returns a pointer to a Boolean variable $\svalue$.
It is expected that if a non-faulty process calls S-Broadcast for a value $v$ then all non-faulty processes
call S-Broadcast for value $v$.
The S-Broadcast abstraction ensures the following properties.

\begin{itemize}
\item S-Termination. The invocation of $\sbroadcast()$ by a non-faulty process terminates.
\item S-Justification. If $\svalue_i$ returned by a call to S-Broadcast with value $v$ at a non-faulty process $p_i$
  has $\svalue_i = \ttrue$, then a non-faulty process has called $\sbroadcast(v, \ttrue)$
  
\item S-Uniformity. If $\svalue_i$ returned by a call to S-Broadcast with value $v$ at a non-faulty process $p_i$
  has $\svalue_i = \ttrue$, then eventually $\svalue_j = \ttrue$ at every non-faulty process $p_j$.

\item S-Obligation. If at least $t+1$ non-faulty processes have called $\sbroadcast(v, \ttrue)$,
  then eventually $\svalue_i = \ttrue$ at every non-faulty process $p_i$.


\end{itemize}

An implementation of S-Broadcast is described in Figure~\ref{algo-SBC}.
Here non-faulty processes call S-Broadcast with an input $v$ and a Boolean $\shouldbroadcast$.
On line~\ref{SBC-01} non-faulty processes initiate the $\svalue$ variable to false.
If $\shouldbroadcast$ is true, then non-faulty processes broadcast a message $\msg,\sval(v)$.
A pointer to $\svalue$ is then returned, note that $\svalue$ may become true at a later point in time.

Lines~\ref{SBC-04}-\ref{SBC-09} describe what happens when a message $\msg,\sval(v)$ is received with
a value $v$ equal to the input given during the invocation of S-Broadcast.
If message $\msg,\sval(v)$ has been received from ($t+1$) distinct processes then the process broadcasts the same message
if it has not already done so (lines~\ref{SBC-04}-\ref{SBC-05}) (i.e. it echos the message).
Next, if messages $\msg,\sval(v)$ has been received from ($2t+1$) distinct processes then the process sets
$\svalue$ to true (lines~\ref{SBC-07}-\ref{SBC-08}).

\begin{lemma}
  Figure~\ref{algo-SBC} satisfies the SBC-Broadcast abstraction.
  \label{lem:sbc}
\end{lemma}
\begin{proof}
  Proof of S-Termination: Figure~\ref{algo-SBC} has no blocking operations.
  Proof of S-Justification: If no non-faulty process calls SBC-Broadcast with
  $\shouldbroadcast = \ttrue$ then no non-faulty process broadcasts $\msg\sval(v)$ on line~\ref{SBC-02}.
  Given $t$, no non-faulty process receives ($t+1$) $\msg\sval(v)$ message from distinct processes
  and no non-faulty process broadcasts $\msg\sval(v)$ on line~\ref{SBC-05}.
  From this and given $2t+1 > t+1$ no non-faulty process sets $\svalue$ to true on line~\ref{SBC-08}.
  Proof of S-Uniformity: If a non-faulty process sets $\svalue$ to true on line~\ref{SBC-08} then it has received
  ($2t+1$) $\msg\sval(v)$ messages from distinct processes. From this and given $t < n/3$ all non-faulty processes
  will receive at least ($t+1$) $\msg\sval(v)$ messages from distinct processes and broadcast $\msg\sval(v)$ if not
  already done (line~\ref{SBC-05}). From this all non-faulty processes will receive at least ($2t+1$)
  messages from distinct processes and set $\svalue \leftarrow \ttrue$ (line~\ref{SBC-08}).
  Proof of S-Obligation: If at least ($t+1$) non-faulty processes call S-Broadcast with
  $\shouldbroadcast = \ttrue$ then these processes broadcast $\msg\sval(v)$ on line~\ref{SBC-02} (if not already done)
  and all non-faulty processes receive at least ($t+1$) messages from distinct processes.
  The proof follows using the same arguments as S-Uniformity.
\end{proof}

\subsection{A Safe and Live Consensus Algorithm in ${\BAMP}[t<n/3,SCC]$.}

This section presents a Binary Byzantine Consensus algorithm using
S-Broadcast in addition to a strong ($t+1$) common coin.

\paragraph{Message types.}
The following message types are used by the consensus.
\begin{itemize}
\item $\aux{r}{s}$. An {\sc aux} message contains a round number $r$ and a binary value $s$.
\end{itemize}

\paragraph{Local variables.}
The following local variables are used at each process.
\begin{itemize}
\item $r_i$. The current round at process $i$.
\item $s_i$. The binary value returned by the most recent call to {\sf random}() at process $i$.
\item $\supportcoin_i$. A Boolean variable that indicates if the current estimate at processes $i$ is
  equal to the value of the most recent coin flip.
\item $\binptr_i$. A list of two pointers to Boolean variables returned by calls to S-Broadcast at process $i$,
  where $\binptr_i[0]$ represents the pointer returned by a call to $\sbroadcast(0, )$
  and $\binptr_i[1]$ represents the pointer returned by a call to $\sbroadcast(1, )$.
\end{itemize}

\begin{figure*}[ht!]
\centering{
\fbox{
\begin{minipage}[t]{150mm}
\footnotesize
\renewcommand{\baselinestretch}{2.5}
\resetline
\begin{tabbing}
aaaA\=aaA\=aaaA\=aaaaaaaaaA\kill

{\bf opera}\={\bf tion} $\propose{v_i}$ {\bf is} \\

\line{BBC4-01} \> $s_i \leftarrow \neg v_i$; $r_i \leftarrow 0$; $\supportcoin_i = \tfalse$; \\



\line{BBC4-02} \> $\binptr_i[s_i] \leftarrow \sbroadcast$ $\est[1](s_i, \tfalse);$ \\
{\it \scriptsize  \hfill \color{gray}{// the Boolean pointed to by $\binptr_i[s_i]$ has not necessarily obtained its final value at this point}} \\

{\bf repeat forever} \\


\line{BBC4-03} \> $r_i \leftarrow r_i + 1$; \\

\line{BBC4-04} \> $\binptr_i[\neg s_i] \leftarrow \sbroadcast$ $\est[r_i](\neg s_i, \neg \supportcoin_i);$ \\
{\it \scriptsize  \hfill \color{gray}{// the Boolean pointed to by $\binptr_i[\neg s_i]$ has not necessarily obtained its final value at this point}} \\

\line{BBC4-05} \> \wait ($\binptr_i[0] = \ttrue \vee \binptr_i[1] = \ttrue$); \\

\line{BBC4-06} \> {\bf case} ($\supportcoin_i = \ttrue$) {\bf then} $w \leftarrow s_i$ \\

\line{BBC4-07} \>\> $~~~$ ($\binptr_i[0] = \ttrue$) $~~$ {\bf then} $w \leftarrow 0$ \\

\line{BBC4-08} \>\> $~~~$ ($\binptr_i[1] = \ttrue$) $~~$ {\bf then} $w \leftarrow 1$ \\

\line{BBC4-09} \> {\bf end case}; \\

\line{BBC4-10} \> $\broadcast$ $\aux{r_i}{w}$; \\

\line{BBC4-11} \> \wait ($\exists$ a set $\view_i[r_i]$ such that (i) for $v \in \view_i[r_i]$, $\binptr_i[v] = \ttrue$  \\
\>\>   $~~~~~~~$ (ii) its values come from messages $\aux{r_i}{}$ received from $(n-t)$ distinct processes); \\


\line{BBC4-12} \> $s_i \leftarrow$ {\sf random}(); \\

\line{BBC4-13} \> {\bf case} ($\view_i[r_i]= \{s_i\}$) $~$ {\bf then} $\supportcoin_i = \ttrue$; $\decide(s_i)$ if not done yet; \\

\line{BBC4-14} \>\> $~~~$ ($\view_i[r_i]= \{0,1\}$) {\bf then} $\supportcoin_i = \ttrue$ \\

\line{BBC4-15} \>\> $~~~~~~~~~~~~~~~~~~~~~~~~~~~~~$ {\bf else} $~$ $\supportcoin_i = \tfalse$ \\

\line{BBC4-16} \> {\bf end case}; \\

{\bf end repeat}



\end{tabbing}
\normalsize
\end{minipage}
}
\caption{ An algorithm implementing binary consensus in ${\BAMP}[t<n/3,SCC]$.}
\label{algo-BBC4} 
\vspace{-1em}
}
\end{figure*}

\paragraph{Algorithm description.}
Figure~\ref{algo-BBC4} presents the algorithm.
Before describing the algorithm line by line some key points are introduced.
First note that, given lines~\ref{BBC4-12}-\ref{BBC4-13}, only the value of the coin
can be decided in a round.
Second, given line~\ref{BBC4-04}, only the negation of the coin from the previous round is S-Broadcast
in all rounds following the first round.
These two points help support the intuition of the design of the algorithm as follows
(assume it is given that the output of the coin in a round $r$ is $s$):
(i) given that $\neg s$ could not have been decided in round $r$, if $s$ was a possible valid decision
in round $r$ then it remains so in round $r+1$, so a process can immediately support $s$ in round $r+1$ and
(ii) given that $s$ could have been decided in round $r$, $\neg s$ must be checked in round
$r+1$ to see if it is still a valid value to decide.
Using this and various thresholds
(including the $t+1$ threshold of the coin) and broadcasts liveness and correctness is then ensured.

The algorithm is now briefly described.
Non-faulty processes call $\propose$ with their initial proposal $v$.
Line~\ref{BBC4-01} initiates the local variables so that on lines~\ref{BBC4-02}-\ref{BBC4-04}
non-faulty processes call $\binptr[\neg v] \leftarrow \sbroadcast$ $\est[1](\neg v, \tfalse)$
and $\binptr[v] \leftarrow \sbroadcast$ $\est[1](v, \ttrue)$,
(i.e. in round $1$ they call the S-Broadcast abstraction for both $0$ and $1$, but with $\supportcoin = \ttrue$ only
for their initial proposal).
Given $t < n/3$ and S-Obligation at at least one of the $\binptr$ variables point to a Boolean variable
that becomes true and by S-Justification was proposed by a non-faulty process in the first round.

Lines~\ref{BBC4-03}-\ref{BBC4-16} are repeated for each round.
The round starts by incrementing the round variable (line~\ref{BBC4-03}.
Non-faulty processes then make a call to $\sbroadcast$ on line~\ref{BBC4-04} with the negation
of the coin from the previous round and the negation of the Boolean variable $\supportcoin$
set on the previous round as inputs.
This ensures that the input $\shouldbroadcast$ to $\sbroadcast$ is only true
if the process decided not to support the output coin from the previous round.
Given that the coin could have been decided in the previous round,
this call to the S-Broadcast abstraction is to check if the negation
of the coin remains a possible value to decide in round $r+1$.
The $\binptr$ pointer to the Boolean variable corresponding to the coin is not changed as its state
of being a value that can be decided remains the same as the previous round.
Non-faulty processes then wait on line~\ref{BBC4-05} until at least one of the $\binptr$ variables point
to a $\ttrue$ Boolean value.
On lines~\ref{BBC4-06}-\ref{BBC4-09} processes compute the value $w$ that they broadcast in an $\aux{r}$
message on line~\ref{BBC4-10}.
This value $w$ is either the value of the coin from the previous round if $\supportcoin = \ttrue$ or a binary
value for which $\binptr[w] = \ttrue$.
The set $\view[r]$ is then computed on line~\ref{BBC4-11} from the the values included with $\aux{r}$ messages
received from ($n-t$) distinct processes for which the corresponding $\binptr$ variables point to a true Boolean.
The strong ($t+1$) common coin is then computed on line~\ref{BBC4-14}.
If $view[r]$ contains a single value equal to the output of the coin then that value is decided on line~\ref{BBC4-13}
and $\supportcoin$ is set to true.
Otherwise if $view[r] = \{0,1\}$, $\supportcoin$ is set to true and no value is decided (line~\ref{BBC4-14}).
Otherwise $\supportcoin$ is set to false and no value is decided (line~\ref{BBC4-15}).
Given $t < n/3$, if a non-faulty process decided then any set of ($n-t$) $\aux{r}$ messages contain
at least one message supporting the coin, thus ensuring all non-faulty processes set $\supportcoin \leftarrow true$.
The round is then complete and the next round is started.

\subsubsection{Proofs.}

This section shows that the algorithm of Figure~\ref{algo-BBC4} solves Binary Byzantine consensus
and terminates in an expected constant number of rounds.

\begin{lemma}
  Non-faulty processes complete each round.
  \label{lem:bbc4-round}
\end{lemma}
\begin{proof}
  Each non-faulty process proposes an initial binary value $v$.
  By S-Termination all non-faulty processes will reach line~\ref{BBC4-05} calling S-Broadcast with both $0$ and $1$
  on lines~\ref{BBC4-02},\ref{BBC4-04}.
  On line~\ref{BBC4-03} each non-faulty process calls $\binptr[\neg v] \leftarrow \sbroadcast$ $\est[1](v, \ttrue)$
  (note that here the variables are replaced to show that non-faulty processes are calling S-Broadcast with their proposal
  and true as input).  
  Given that $t < n/3$ at least one of the S-broadcast instances will be called with
  $\shouldbroadcast = \ttrue$ for a single binary value by at least $t+1$ non-faulty processes, and
  by S-Obligation the condition on line~\ref{BBC4-06} will be satisfied.
  It follows that either line~\ref{BBC4-08} or line~\ref{BBC4-09} (or both) be satisfied at all non-faulty processes
  and they will broadcast $\aux{1}{w}$ where $\binptr[w] = \ttrue$.
  Now given S-Obligation all non-faulty processes will receive at least $(n-t)$ $\aux{1}{}$ messages with
  values $w$ that satisfy $\binptr[w] = \ttrue$.
  All non-faulty processes will then call ${\sf random}()$ and the value of the coin will be output.
  The case on lines~\ref{BBC4-13}-\ref{BBC4-15} will then be completed and non-faulty processes will
  continue onto round 2.

  Now consider by induction that all non-faulty processes have completed round $r$.
  Let the output of the coin in round $r$ be binary value $s$.
  If a non-faulty process sets $\supportcoin = true$ in round $r$ it must have had
  $\binvalues[s] = \ttrue$ (lines~\ref{BBC4-13},\ref{BBC4-14}), and given the pointer $\binvalues[s]$
  is not modified in round $r+1$, $\binvalues[s] = \ttrue$ at all non-faulty processes in round $r+1$ by S-Obligation.
  Call this note property (a).
  
  Consider the following two cases representing the values of $\supportcoin$ set in round $r$ at non-faulty processes:
  \begin{itemize}
  \item At least $t+1$ non-faulty processes set $\supportcoin \leftarrow \tfalse$ in round $r$ on line~\ref{BBC4-15}.
    Now given the strong property of the coin, these processes will then call
    $\binptr[\neg s] \leftarrow \sbroadcast \est[r+1](\neg s, \ttrue)$ and by S-Obligation eventually
    $\binptr[\neg s] = \ttrue$ at all non-faulty processes in round $r+1$ and the wait condition on line~\ref{BBC4-05}
    will be satisfied.
    From this and by (a) all non-faulty processes will broadcast messages $\aux{r+1}{v}$ where $\binptr[v] = \ttrue$
    on line~\ref{BBC4-10} which will ensure the wait on line~\ref{BBC4-11} will be completed and all non-faulty
    processes will continue to round $r+2$.
  \item Otherwise less than $t+1$ non-faulty processes set $supportcoin \leftarrow \tfalse$ in round $r$ on line~\ref{BBC4-15}.
    In this case less than $t+1$ non-faulty processes will call $\binptr[\neg s] \leftarrow \sbroadcast \est[r+1](\neg s, \ttrue)$
    and by S-Justification $\binptr[\neg s] = \tfalse$ in round $r+1$.
    Furthermore given $t < n/3$ and by (a) $\binvalues[s] = \ttrue$ at all non-faulty processes, all non-faulty processes will set
    $w \leftarrow s$ on lines~\ref{BBC4-06}-\ref{BBC4-08} and broadcast $\aux{r+1}{s}$ on line~\ref{BBC4-10}.
    The wait on line~\ref{BBC4-10} will then be satisfied and all non-faulty processes will continue to round $r+2$.
  \end{itemize}
  This completes the proof by induction.
\end{proof}

\begin{lemma}
  If a non-faulty process decides a binary value $v$ in a round $r$, then at every non-faulty process
  (i) in every round $\geq r$, $\binptr[v] = \ttrue$ (eventually) and no non-faulty process decides $\neg v$ and
  (ii) in every round $> r$, $\binptr[\neg v] = \tfalse$.
  \label{lem:bbc4-binptrsame}
\end{lemma}
\begin{proof}
  Let $r$ be the first round where a non-faulty process decides and $v$ be the binary value
  it decides.
  For this to happen the following must be true: $\binptr[v] = \ttrue$ by line~\ref{BBC4-11}
  and $s = v$ by line~\ref{BBC4-13}. Additionally, by the strong property of the coin, and by S-Obligation
  these will be true at all non-faulty processes.
  Furthermore, by line~\ref{BBC4-11} a process that decides must have received $(n-t)$ $\aux{r}{v}$ messages from
  ($n-t$) distinct processes. Now given that $t < n/3$, any set of ($n-t$) $\aux{r}{}$ messages
  from distinct processes must contain at least one $\aux{r}{v}$ message.
  Thus by lines~\ref{BBC4-13}-\ref{BBC4-14} all non-faulty processes will set $\supportcoin = \ttrue$
  and no non-faulty process will decide $\neg v$ in round $r$.
  Now in round $r+1$ all non-faulty processes will call $\binptr[\neg v] \leftarrow \sbroadcast$ $\est[r+1](\neg v, \tfalse)$
  and by S-Justification $\binptr[\neg v] = \tfalse$ at all non-faulty processes.
  From this non-faulty processes will compute $\view[r+1] = \{v\}$ on line~\ref{BBC4-11}.
  Furthermore given $\binptr[v]$ is not modified from round $r$, $\binptr[v] = \ttrue$ in round $r+1$.

  Now consider by induction that the lemma is true in a round $\rho$.
  First note that given $\binptr[v] = \ttrue$ and $\binptr[\neg v] = \tfalse$ in round $\rho$,
  $\view[\rho] = \{v\}$ at all non-faulty processes by lines~\ref{BBC4-10}-\ref{BBC4-11}.
  Consider the two possible cases for the output of the coin in round $\rho$:
  \begin{itemize}
  \item The output of {\sf random}() in round $\rho$ is $v$. Here the case on line~\ref{BBC4-13}
    is satisfied and all non-faulty processes set $\supportcoin = \ttrue$ and $\decide(v)$ if not done already.
    In round $\rho+1$ all non-faulty processes call $\binptr[\neg v] \leftarrow \sbroadcast$ $\est[\rho+1](\neg v, \tfalse)$
    and by S-Justification $\binptr[\neg v] = \tfalse$ at all non-faulty processes.
    Furthermore given $\binptr[v]$ is not modified from round $\rho$, $\binptr[v] = \ttrue$ in round $\rho+1$.
    From this non-faulty processes compute $\view[\rho+1] = \{v\}$ on line~\ref{BBC4-11} and do not decide $\neg v$.
  \item Otherwise the output of {\sf random}() in round $\rho$ is $\neg v$. Here the case on line~\ref{BBC4-15}
    is satisfied and all non-faulty processes set $\supportcoin \leftarrow \tfalse$.
    In round $\rho+1$ all non-faulty processes call $\binptr[v] \leftarrow \sbroadcast$ $\est[\rho+1](v, \ttrue)$
    and by S-Obligation $\binptr[v] = \ttrue$ (eventually) at all non-faulty processes.
    Furthermore given $\binptr[\neg v]$ is not modified from round $\rho$, $\binptr[\neg v] = \tfalse$ in round $\rho+1$.
    From this non-faulty processes compute $\view[\rho+1] = \{v\}$ on line~\ref{BBC4-11} and do not decide $\neg v$.
  \end{itemize}
  This completes the proof by induction.
\end{proof}

\begin{lemma}
  No two non-faulty processes decide different values.
  \label{lem:bbc4-decsame}
\end{lemma}
\begin{proof}
  This follows directly from~\ref{lem:bbc4-binptrsame}.
\end{proof}

\begin{lemma}
  Non-faulty processes only call $\sbroadcast(v, \ttrue)$ if $v$ was proposed by a non-faulty process.
  \label{lem:bbc4-nfpropose}
\end{lemma}
\begin{proof}
In round $1$ this is true as non-faulty processes only call $\sbroadcast(v, \ttrue)$ for their proposal $v$.
Now assume the lemma holds true in round $r$.
In round $r+1$ processes only call $\sbroadcast(v, \ttrue)$ if it set $\supportcoin$ to false
on line~\ref{BBC4-15} in round $r$.
For this to happen by line~\ref{BBC4-11} $\binptr[v]$ must be true in round $r$.
The proof then follows by S-Justification and induction.  
\end{proof}

\begin{lemma}
  Non-faulty processes only decide binary values proposed by non-faulty processes.
  \label{lem:bbc4-decpropose}
\end{lemma}
\begin{proof}
  By Lemma~\ref{lem:bbc4-nfpropose} and S-Justification the pointers $\binptr$ at non-faulty processes
  only have value $true$ if $\sbroadcast(v, \ttrue)$ was called by a non-faulty process with their
  corresponding binary value $v$.
  By lines~\ref{BBC4-11} and~\ref{BBC4-13} a non-faulty process can only decide a binary value $v$
  if $\binptr[v] = \ttrue$, the proof follows.
\end{proof}

\begin{lemma}
  Given a binary value $s$ as the output of the call to ${\sf random}()$ in a round $r$.
  If all non-faulty processes start round $r+1$ with $\supportcoin = \ttrue$,
  then all non-faulty processes decide (if they have not done so already)
  in the first round $\rho > r$ where the output of ${\sf random}()$ in round $\rho$ is $s$.
  \label{lem:bbc4-allsupportcoin}
\end{lemma}
\begin{proof}
  Given by definition that $\supportcoin = \ttrue$ at all non-faulty processes in round $r+1$, they must have executed
  either line~\ref{BBC4-13} or~\ref{BBC4-14} in round $r$ and must have $\binptr[s] = \ttrue$.
  Given $\binptr[s]$ is not modified in round $r$ it remains true in round $r+1$.
  Now all non-faulty processes call $\binptr[\neg s] \leftarrow \sbroadcast$ $\est[r+1](\neg s, \tfalse)$
  on line~\ref{BBC4-04} and by S-Justification $\binptr[\neg s] = \tfalse$ at all non-faulty processes in round $r+1$.
  Notice now that $\binptr[\neg s] = \tfalse$ and $\binptr[s] = \ttrue$ at all non-faulty processes, call this
  state (a).
  
  From state (a) in round $r+1$ all non-faulty processes will
  broadcast $\aux{r+1}{s}$ on line~\ref{BBC4-10} and have $\view[r+1] = \{v\}$
  from line~\ref{BBC4-11}.
  
  Now if the call to ${\sf random}()$ returns $s$ then all non-faulty
  processes decide on line~\ref{BBC4-13}.

  Otherwise the call ${\sf random}()$ returns $\neg s$ in round $r+1$
  and all non-faulty processes set $\supportcoin=\tfalse$ on line~\ref{BBC4-15}.
  In round $r+2$ $\binptr[\neg s] = \tfalse$ at all non-faulty processes given it is not changed from round $r+1$.
  All non-faulty processes then call $\binptr[s] \leftarrow \sbroadcast$ $\est[r+2](\neg s, \ttrue)$
  on line~\ref{BBC4-04} and by S-Obligation $\binptr[s] = \ttrue$ at all non-faulty processes in round $r+2$.
  Notice that the non-faulty processes are once again in state (a) except in round $r+2$,
  the proof then follows.
\end{proof}

\begin{lemma}
  Non-faulty processes decide in expected $O(1)$ rounds.
  \label{lem:bbc4-decconstant}
\end{lemma}
\begin{proof}
  Consider a round $r$.
  Given the $t+1$ property of the coin, at least $t+1$ non-faulty processes have reached line~\ref{BBC4-12}
  before the output of ${\sf random}()$ is revealed.
  Consider the following two cases when the $t+1$th non-faulty process reaches this line:
  \begin{itemize}
  \item At least one of these non-faulty processes has computed $\view[r] \leftarrow \{v\}$ on line~\ref{BBC4-11}, i.e.
    $\view[r]$ contains a single binary value $v$.
    In this case if the output of the call to {\sf random}() is $v$ then the process will decide $v$
    on line~\ref{BBC4-13} and decide in the next round where the output of the coin is $v$ (if not done already)
    by Lemmas~\ref{lem:bbc4-binptrsame} and~\ref{lem:bbc4-allsupportcoin}.
  \item Otherwise given $t < n/3$
    all $t+1$ of these non-faulty processes have set $\view[r] \leftarrow \{0,1\}$ on line~\ref{BBC4-11}.
    Let the output of the the call to {\sf random}() be $s$ in round $r$.
    Note that the wait on line~\ref{BBC4-11} and S-Uniformity ensure $\binptr[s] = \ttrue$ at all non-faulty
    processes in round $r$ and given the pointer is not changed in round $r+1$, this remains true in round $r+1$.
    Now on line~\ref{BBC4-14} the same $t+1$ non-faulty processes set $\supportcoin \leftarrow \ttrue$ 
    and in round $r+1$ broadcast $\aux{r+1}{s}$.
    Given $t < n/3$ all sets of ($n-t$) distinct $\aux{r+1}{}$ messages will contain at least
    one $\aux{r+1}{s}$ message.
    It then follows that $s \in \view[r+1]$ at all non-faulty processes.
    Now if the output of {\sf random}() is $s$ in round $r+1$ then all non-faulty processes will set
    $\supportcoin_i \leftarrow \ttrue$ on lines~\ref{BBC4-13} or line~\ref{BBC4-14} and by Lemma~\ref{lem:bbc4-allsupportcoin}
    all non-faulty processes will decide by the next round where the output of the call to ${\sf random}()$ is $s$.
  \end{itemize}
  In both cases in round $r$ non-faulty processes reach a state where they will reach a decision with probability
  of at least $1/2$ by the strong property of the coin, or decision
  is ensured with probability $1 - \prod_{r=1}^{\infty} 1/2 = 1$.
  From this, the expected number of rounds to reach a state from which a decision is ensured is
  $\sum_{r=1}^{\infty}r\frac{1}{2^r} = 2$, and by Lemmas~\ref{lem:bbc4-allsupportcoin} and~\ref{lem:bbc4-binptrsame}
  all processes will decide by the next round where the coin flip results in the same value,
  i.e. another expected $2$ rounds.
  
\end{proof}

\begin{theorem}
  The algorithm presented in Figure~\ref{algo-BBC4} solves the Binary consensus problem in ${\BAMP}[t<n/3,CC]$.
\end{theorem}
\begin{proof}
  First recall the definition of Binary Byzantine Consensus.
  \begin{itemize}
  \item BBC-Termination. Every non-faulty process eventually decides on a value.
  \item BBC-Agreement.   No two non-faulty processes decide on different values.
  \item BBC-Validity.  If all non-faulty processes propose the same value, no
    other value can be decided.
  \end{itemize}
  BBC-Termination is ensured by Lemma~\ref{lem:bbc4-decconstant}.
  BBC-Agreement and BBC-Validity are ensured by Lemmas~\ref{lem:bbc4-decsame} and~\ref{lem:bbc4-decpropose} respectively.
\end{proof}

\paragraph{Message broadcasts.}
The S-Broadcast consists of at most $1$ message broadcast.
The first round of the consensus algorithm consists of $2$ instances of S-Broadcast (where non-faulty
processes perform a normal broadcast within at least one of these), followed by a normal message broadcast,
or $2$ to $3$ message broadcasts.
All following rounds consist of a call to S-Broadcast followed by a normal message broadcast,
or $1$ to $2$ message broadcasts.

\subsection{Optimization for the Safe and Live Consensus Algorithm in ${\BAMP}[t<n/3,CC]$ from Figure~\ref{algo-BBC}.}
\label{sec:opt}
Notice that the strong  ($t+1$) common coin algorithm from Figure~\ref{algo-BBC4} uses
the fact that a binary value may remain valid from the previous round in order to only perform a single S-Broadcast
per round in rounds following the first round.
A similar technique can be used to reduce the number of broadcasts used by the SBV-Abstractions of the algorithm
in Figure~\ref{algo-BBC}.
Consider the binary value $v$ ($v \neg \bot)$ from $\view[r,2]$ which may be set to non-faulty processes' estimates
on lines~\ref{BBC4-12}-\ref{BBC4-13}.
By Lemma~\ref{lem:binall} and SBV-Binvalues this can only be a single binary value and if such a value
exists then all non-faulty processes have it in their $\binvalues$ returned from the call to $\sbvbroadcast$ $\stage[r,1]$.
With this, non-faulty processes can skip the initial broadcast of the value in the following round
and broadcast it directly within an ${\sc aux}$ message in $\sbvbroadcast$ $\stage[r+1,0]$, where
the $\binvalues$ returned by $\sbvbroadcast$ $\stage[r,1]$ ensures its validity.
Furthermore notice that this does not affect termination as this binary value is only used in the termination proof
when it is equal to the value output by the coin at all non-faulty processes,
in which case all non-faulty processes take this value as the estimate regardless.

This optimization then can be further applied to $\bot$ when input to the $\sbvbroadcast$ $\stage[r,1]$ on line~\ref{BBC4-10}.
Here processes can immediately broadcast $\bot$ in the {\sc aux} message, using the fact that all
non-faulty processes will $\binvalues[r] = \{0,1\}$ (the $\binvalues$ returned by $\sbvbroadcast$ $\stage[r,0]$) as support.
Again termination is not effected using a similar argument as before.

This optimization reduces the number of message broadcast performed by the SBV-Broadcast abstraction from between $2$ and $3$
to between $1$ and $2$.
Notice that if $\sbvbroadcast$ $\stage[r,1]$ only uses a single broadcast, then $\sbvbroadcast$ $\stage[r,2]$ will
use $2$ broadcasts, as the single broadcast of $\sbvbroadcast$ $\stage[r,2]$ can only be used for $\bot$
which requires $\sbvbroadcast$ $\stage[r,1]$ to use $2$ broadcasts.
Thus the number of message broadcasts performed per round becomes $4$ to $5$.

Note that a similar opimization could be applied to the algorithm of~\cite{CGLR18} that uses a
similar construction but relies on a synchrony assumption for termination.


\begin{thebibliography}{99}
\footnotesize{  

\bibitem{A03}
James Aspnes. Randomized protocols for asynchronous consensus.
{\it Distrib. Comput.}, 16(2-3):165-175, September 2003.

\bibitem{BO83}
Michael Ben-Or. Another advantage of free choice (extended abstract):
Completely asynchronous agreement protocols. {\it In Proceedings of the
Second Annual ACM Symposium on Principles of Distributed Computing}, PODC '83, pages 27-30, 1983

\bibitem{BG93}
Berman P. and Garay J.A., Randomized distributed agreement revisited.
{\it 33rd Annual Int'l Symposium on Fault-Tolerant Computing (FTCS' 93)},
IEEE Computer Press, pp. 412-419, 1993.

\bibitem{BSA14}
Alyson Bessani, Joao Sousa, and Eduardo E. P. Alchieri. State machine
replication for the masses with bft-smart. {\it In 2014 44th Annual IEEE/IFIP
International Conference on Dependable Systems and Networks}, pages
355-362, June 2014.


\bibitem{B87}
Gabriel Bracha. An o(log n) expected rounds randomized byzantine
generals protocol. {\it J. ACM}, 34(4):910-920, October 1987

\bibitem{BT83}
Gabriel Bracha and Sam Toueg. Asynchronous consensus and byzantine protocols in faulty environments. Technical Report TR83-559,
Cornell University, 1983.




\bibitem{CKS05}
Christian Cachin, Klaus Kursawe, and Victor Shoup. Random oracles
in constantinople: Practical asynchronous byzantine agreement using
cryptography. {\it Journal of Cryptology}, 18(3):219-246, 2005.

\bibitem{CR93}
Ran Canetti and Tal Rabin. Fast asynchronous byzantine agreement
with optimal resilience. In Proceedings of the Twenty-fifth Annual ACM
Symposium on Theory of Computing, {\it STOC '93}, pages 42-51, 1993.

\bibitem{CL02}
 Miguel Castro and Barbara Liskov. Practical byzantine fault tolerance
and proactive recovery. {\it ACM Trans. Comput}. Syst., 20(4):398-461,
November 2002.

\bibitem{C20}
Tyler Crain. A Simple and Efficient Asynchronous Randomized Binary Byzantine Consensus Algorithm.
{\it arXiv preprint arXiv:2002.04393}, 2020.
  
\bibitem{C220}
Tyler Crain. A Simple and Efficient Binary Byzantine Consensus Algorithm
using Cryptography and Partial Synchrony. {\it arXiv preprint arXiv:2001.07867}, 2020.

\bibitem{CGLR18}
Tyler Crain, Vincent Gramoli, Mikel Larrea, and Michel Raynal.
Dbft: Efficient leaderless byzantine consensus and its applications to
blockchains. {\it In Proceedings of the 17th IEEE International Symposium
on Network Computing and Applications (NCA'18)}. IEEE, 2018.




\bibitem{DDS87}
Danny Dolev, Cynthia Dwork, and Larry Stockmeyer. On the minimal
synchronism needed for distributed consensus. {\it J. ACM}, 34(1):77-97,
January 1987.


\bibitem{DLS88}
Cynthia Dwork, Nancy A. Lynch, and Larry J. Stockmeyer. Consensus
in the presence of partial synchrony. {\it J. ACM}, 35(2):288-323, 1988.

\bibitem{FM97}
PESECH FELDMAN and SILVIO Micali. An optimal probabilistic
protocol for synchronous byzantine agreement. {\it SIAM J. Computing},
26(4):873-933, 1997.


\bibitem{FL82}
Fischer M.J. and Lynch N.A.,
A lower bound for the time to assure interactive consistency.
{\em Information Processing Letters}, 14(4):183-186 (1982)


\bibitem{FLP85}
Fischer M.J., Lynch N.A.,  and Paterson M.S.,
Impossibility of distributed consensus with one faulty process.
{\em Journal of the ACM}, 32(2):374-382 (1985)

\bibitem{FMR07}
Friedman R., Most\'efaoui A., Rajsbaum S., and Raynal M., Distributed agreement problems and their
connection with error-correcting codes. {\it IEEE Transactions on Computers}, 56(7):865-875, 2007. 

\bibitem{FP90}
Oded Goldreich and Erez Petrank. The best of both worlds: Guaranteeing termination in fast randomized byzantine agreement protocols.
{\it Inf. Process. Lett.}, 36(1):45-49, 1990.

  

\bibitem{KS16}
  Valerie King and Jared Saia. Byzantine agreement in expected polynomial time. {\it J. ACM}, 63(2):13, 2016.


\bibitem{LSP82}
Leslie Lamport, Robert Shostak, and Marshall Pease. The byzantine
generals problem. {\it ACM Trans. Program. Lang. Syst.}, 4(3):382-401, July
1982.

\bibitem{LVCQV16}
Shengyun Liu, Paolo Viotti, Christian Cachin, Vivien Qu\'ema, and
Marko Vukolic. XFT: practical fault tolerance beyond crashes. {\it In 12th
USENIX Symposium on Operating Systems Design and Implementation,
OSDI 2016}, Savannah, GA, USA, November 2-4, 2016., pages 485-500,
2016.

\bibitem{M18}
Ethan MacBrough. Cobalt: BFT Governance in Open Networks.
{\it arXiv preprint arXiv:1802.07240}, 2018.

\bibitem{MA06}
Jean-Philippe Martin and Lorenzo Alvisi. Fast byzantine consensus.
{\it IEEE Trans. Dependable Sec. Comput.}, 3(3):202-215, 2006.

\bibitem{MMR14}
Achour Most\'efaoui, Hamouma Moumen, and Michel Raynal.
Signature-free asynchronous byzantine consensus with $T < N /3$ and
$O(N^2)$ messages. {\it In Proceedings of the 2014 ACM Symposium on Principles of Distributed Computing, PODC '14},
pages 2-9, New York, NY,
USA, 2014. ACM.

\bibitem{MMR15}
Achour Mostéfaoui, Hamouma Moumen, and Michel Raynal.
Signature-Free Asynchronous Binary Byzantine Consensus with t < n/3, O(n2) Messages,
and O(1) Expected Time. {\it J. ACM 62, 4}. Article 31. 2015.


\bibitem{MR17}
Achour Most\'efaoui and Michel Raynal. Signature-free asynchronous byzantine systems: from multivalued to binary consensus with
$t < n/3$, $O(n^2)$ messages, and constant time. {\it Acta Informatica, 2017}.
Accepted: 19 April 2016

\bibitem{MRT00}
Achour Most\'efaoui, Michel Raynal, and Fr\'ed\'eric Tronel. From binary
consensus to multivalued consensus in asynchronous message-passing
systems. {\it Inf. Process. Lett., 73(5-6):207-212}, March 2000.

\bibitem{NCV05}
N. F. Neves, M. Correia, and P. Verissimo. Solving vector consensus
with a wormhole. {\it IEEE Trans. on Parallel and Distributed Systems},
16(2):1120-1131, 2005.

\bibitem{PCR14}
Arpita Patra, Ashish Choudhury, and C. Pandu Rangan. Asynchronous
byzantine agreement with optimal resilience. {\it Distributed Computing},
27(2):111-146, 2014.

\bibitem{PSL80}
M. Pease, R. Shostak, and L. Lamport. Reaching agreement in the
presence of faults. {\it J. ACM}, 27(2):228-234, April 1980

\bibitem{R83}
Michael O. Rabin. Randomized byzantine generals. {\it In Proceedings of
the 24th Annual Symposium on Foundations of Computer Science, SFCS
'83}, pages 403-409, 1983.


\bibitem{ST87}
Srikanth T.K. and Toueg S., Simulating authenticated broadcasts to derive simple fault-tolerant algorithms.
{\it Distributed Computing}, 2:80-94, 1987.


\bibitem{T84}
Sam Toueg. Randomized byzantine agreements. {\it In Proceedings of the
Third Annual ACM Symposium on Principles of Distributed Computing,
PODC '84}, pages 163-178, 1984.

\bibitem{TC84}
  Russell Turpin and Brian A. Coan. Extending binary byzantine agreement to multivalued byzantine agreement.
  {\it Inf. Process. Lett.}, 18(2):73-
76, 1984.

\bibitem{ZC09}
Jialin Zhang and Wei Chen. Bounded cost algorithms for multivalued
consensus using binary consensus instances. {\it Information Processing
Letters}, 109(17):1005-1009, 2009.


}
\end{thebibliography}
\end{document}